\documentclass[runningheads]{llncs}
\usepackage[T1]{fontenc}
\usepackage{graphicx}
\usepackage{amsmath}
\usepackage{amssymb}
\usepackage{amsfonts}

\newcommand{\classfont}[1]{\mathbf{#1}}
\newcommand{\clP}{\classfont{P}}
\newcommand{\clPP}{\classfont{PP}}
\newcommand{\clPH}{\classfont{PH}}
\newcommand{\clDRMA}{\classfont{DRMA}}
\newcommand{\clFRMA}{\classfont{FRMA}}
\newcommand{\clParityP}{\oplus\classfont{P}}
\newcommand{\clSharpP}{\#\classfont{P}}
\newcommand{\trT}{\mathcal{T}} %for denoting transcript
\DeclareMathOperator{\expect}{\mathbb{E}}
\renewcommand{\phi}{\varphi}
\DeclareMathOperator{\poly}{poly}

\title{Structural Complexity of Rational Interactive Proofs}

\author{Daniil Musatov\inst{1,2,3}\orcidID{0000-0002-1779-2513} \and Georgii Potapov\inst{1} \orcidID{0009-0002-2670-2206}}
\authorrunning{D. Musatov and G. Potapov}
\institute{Moscow Institute of Physics and Technology, Dolgoprudny, Russia \and
Russian Presidential Academy of National Economy and Public Administration, Moscow, Russia \and
Caucasus Mathematical Center at Adyghe State University, Maykop, Russia\\
}

\begin{document}

\maketitle
\begin{abstract}
    This is the full version of a paper submitted to the Computability in Europe (CiE 2023) conference, with all proofs omitted there.
    
    In 2012 P. D. Azar and S. Micali introduced a new model of interactive proofs, called {\lq\lq Rational Interactive Proofs\rq\rq}. In this model the prover is neither honest nor malicious, but rational in terms of maximizing his expected reward. In this article we explore the connection of this area with classic complexity results.

    In the first part of this article we revise the ties between the counting hierarchy and the hierarchy of constant-round rational proofs. We prove that a polynomial-time machine with oracle access to $\clDRMA[k]$ decides exactly languages in $\clDRMA[k]$, a coincidence unknown for levels of the counting hierarchy.
    
    In the second part we study communication complexity of single-round rational proofs. We show that
    %, as a result of a theorem of Fortnow,
    the class defined by logarithmic-communi\-ca\-tion single-round rational proofs coincides with $\clPP$. We also show that single-round rational protocols that treat problems in $\clParityP$ as black-box samplers of a random variable require at least a linear number of bits of communication.
\end{abstract}

\section{Introduction}

Advances in computing technologies not only change our world, but also pave the way to new areas in theoretical computer science. Thus, in 1960s and 1970s emergence of full-scale computers led to evolvement of complexity theory and design of algorithms. Later, in 1980s and 1990s development of computer networks led to invention of interactive proofs and cryptographic protocols. In 2012 a new field that reflects commercial cloud computing was shaped, namely the field of rational interactive proofs.

In cloud computing, a client sends a computational task to a server and receives some response. How can the client be sure that the response is indeed the result of the requested computation? If the task is an $\mathbf{NP}$ problem, then the server may send a certificate verifiable by the client. But what if the task is harder? The main idea of rational proofs is to create economic incentives for the server to do the computation correctly. To this end, the standard model of interactive proofs is expanded by a polynomially computable probabilistic reward paid by the client to the server. In order to maximize the expectation of this reward, the server reveals the correct answer.

Thus, while in standard interactive-proof and cryptographic models the prover may be either honest or malicious, here we consider a single type of the prover, the economically rational one. At the first glance, this does not change much for interactive proofs: the verifier may just pay 1 dollar for an accepted proof and 0 dollars for a rejected one. But things greatly change when we consider round and communication complexity. In our paper, we consider several concrete questions about them.

\subsection{History and related literature}

The area of interactive proofs was introduced in the works of Babai~\cite{Babai85} (for public random bits) and Goldwasser, Micali and Rackoff~\cite{GMR89} (for private random bits). Then Goldwasser and Sipser showed~\cite{GS86} that a protocol with private bits may be transformed to a $2$ rounds longer protocol with public bits. Babai and Moran demonstrated~\cite{BM88} that the number of rounds in a public-bit protocol may be halved. Together these results prove that any language recognizable by a constant-round interactive proof system is in fact recognizable by a two-round public-bit protocol and thus lies in $\mathbf{AM}$ and hence in the second level of the polynomial hierarchy. The seminal result $\mathbf{IP}=\mathbf{PSPACE}$~\cite{LFKN92,Shamir90,Shen92} showed that a polynomial number of rounds increases the class to $\mathbf{PSPACE}$. The case of multiple provers was also analyzed, and it turned out that the respective class $\mathbf{MIP}$ is equal to $\mathbf{NEXP}$~\cite{ben1988multi,fortnow1994power,babai1991non}.

Meanwhile, Wagner introduced~\cite{Wagner86} the counting hierarchy, an analog of the polynomial hierarchy, where the existential and universal quantifiers are replaced by majority quantifiers. The famous Toda's theorem~\cite{Toda91} implies that the whole polynomial hierarchy lies in the second level of the counting hierarchy.

Interactive proofs and the counting hierarchy meet in the area of rational interactive proofs that was pioneered by Azar and Micali in \cite{azar2012rational}. They introduced the notion of a rational proof, defined the functional class $\mathbf{FRMA}$ and the decision class $\mathbf{DRMA}$ and proved that $\mathbf{DRMA}$ with constant round complexity equals the counting hierarchy, thus showing the difference from $\mathbf{IP}$ with constant rounds. On the other hand, if the number of rounds is polynomial, then the respective class, denoted by $\mathbf{RIP}$, is still equal to $\mathbf{PSPACE}$. 

The subsequent analysis went in several directions. Firstly, the concept of rational proofs was also expanded to the case of multiple provers. Since the payoffs of the provers are mutually dependent, the provers play some game with each other and the result depends on the rules of this game. Chen, McCauley and Singh consider both cooperative~\cite{chen2016rational} and non-cooperative~\cite{chen2019non} approaches and show that the respective classes are not only wider than $\mathbf{RIP}$ but also wider than $\mathbf{MIP}$ (the exact characterization depends on fine details of the model). Thus, rationality expands the class of recognizable languages for several provers.

Secondly, various models with computationally limited prover and/or verifier were considered. Azar and Micali introduce~\cite{azar2013super,azar2014super} super-efficient rational proofs, where the verifier is logarithmically time-bounded and the honest prover need not do more computation than to solve the problem. They characterize in these terms classes $\mathbf{TC}^0$, $\mathbf{P}^{\|\mathbf{NP}}$ and $\mathbf{P}^{\|\mathbf{MA}}$. They also show that the verifier needs only a polynomial budget: deviation from the optimal strategy leads to a considerable loss. Guo et al.~\cite{guo2014singleround} stepped from $\mathbf{TC}^0$ to $\mathbf{NC}^1$ and Chen, McCauley and Singh~\cite{chen2018efficient} improved the bounds for utility gaps. Campanelli and Gennaro constructed a sequentially composable protocol~\cite{campanelli2015sequentially} and yet another space-efficient one~\cite{campanelli2017efficient}. Inasawa and Yasunaga explored~\cite{inasawa2017rational} a variant where both prover and verifier are rational.

\subsection{Exposition of the results}

In this paper we return to the original Azar-Micali framework and explore the relations between the counting and the $\mathbf{DRMA}$ hierarchies. In~\cite{azar2012rational} it is shown that a language recognizable by a $k$-round rational proof lies somewhere between the $k$th and the $(2k+1)$st levels of the counting hierarchy. We try to shed light on whether these bounds are tight. Is it possible that the two hierarchies coincide?

We answer in the following way: either the hierarchies are different, or their coincidence would lead to resolving a long-time open question.

\textbf{Result 1} (Corollary~\ref{crl:p-drma}): $\clP^{\clDRMA[k]} = \clDRMA[k]$.

On the contrary, the equality $\clP^{C_k\clP}=C_k\clP$ is not known: for instance, by Toda's theorem $\mathbf{PH}\subset\mathbf{P}^{\mathbf{PP}}$, but the inclusion $\mathbf{PH}\subset\mathbf{PP}$ is not known. 

The next two results show that the levels of the $\mathbf{DRMA}$ hierarchy relate to each other in a way similar, but not equivalent to the counting classes.

\textbf{Result 2} (Theorem~\ref{thm:pp_with_frma_oracle}): $\mathbf{PP}^{\mathbf{DRMA}[k]}\subset\mathbf{DRMA}[k+1]$. 

We do not know whether this inclusion holds as equality. On the contrary, $\mathbf{PP}^{C_k\mathbf{P}}=C_{k+1}\mathbf{P}$ is just the definition of $C_{k+1}\mathbf{P}$ and always holds.

We do not know either whether the equality $\mathbf{DRMA}[k]=\mathbf{DRMA}[k+1]$ implies $\mathbf{DRMA}[k+1]=\mathbf{DRMA}[k+2]$, but we show the following connection between collapses of the two hierarchies.

\textbf{Result 3} (Corollary~\ref{thm:nice-corollary}):
If $\mathbf{DRMA}[k]=\mathbf{DRMA}[k+1]$, then the counting hierarchy collapses at most at the $(2k+1)$-st level. 

Finally, we estimate the communication complexity of 1-round proofs.

\textbf{Result 4} (Corollary~\ref{crl:pp_eq_drma_log}): $\mathbf{PP}$ is equal to $\mathbf{DRMA}[1, O(\log n)]$, i.e., to $\mathbf{DRMA}$ with one round and $O(\log n)$ sent bits.

Thus the question of whether $\mathbf{PP}=\mathbf{DRMA}[1]$ is equivalent to the question whether a polynomial number of communicated bits expands the class compared to a logarithmic number. We also briefly discuss the connection between $\clParityP$ and $\clDRMA[1]$ and limitations of our methods.

\subsection{Roadmap}

In Sect.~\ref{sect:prelim} we present formal definitions and basic facts about the counting hierarchy and rational proofs. Sect.~\ref{sect:DRMA} and~\ref{sect:comm} give an exposition of our results about the $\mathbf{DRMA}$ hierarchy and about communication complexity, respectively. In Sect.~\ref{sect:concl} we provide concluding remarks and pose some questions for further research.

\section{Preliminaries}\label{sect:prelim}
\subsection{Counting hierarchy}
Firstly we remind some basic definitions concerning the counting hierarchy. 
\begin{definition}
    The class $\clPP$ is the class of languages $A$ for which there exists a language $B\in\clP$ and a polynomial $q$ such that for all $x$
    \[
      x \in A \iff \frac{1}{2^{q(|x|)}}\left|\left\{y \in \{0,1\}^{q(|x|)} : (x, y) \in B\right\}\right| \geq \frac{1}{2}.
    \]
    %where a pair $(\cdot, \cdot)$ can be thought of in any conventional way, e.g. as polynomial-time computable bijective encoding of pairs of strings over a fixed finite alphabet, for which there exist polynomial-time computable projection functions. 
\end{definition}
We will also use an equivalent definition in which the fraction of acceptable certificates $y$ never equals exactly $1/2$. (In order to obtain it, one needs to add 2 new random bits and artificially shift the accepting probability by a tiny amount).

\begin{definition}The relativised class $\clPP^{\classfont{C}}$ is defined in the same way but with $B$ being a language in $\clP^C$ for some $C \in \classfont{C}$.
\end{definition}
\begin{definition}

    The counting hierarchy is defined recursively as follows:
    \begin{itemize}
        \item $C_0\clP = \clP$,
        \item $C_1\clP = \clPP$,
        \item $C_{k+1}\clP = \clPP^{C_k \clP}$ for $k\geq 1$,
        \item $\classfont{CH} = \bigcup_{k \geq 0} C_k\clP$.
    \end{itemize}
\end{definition}
It is clear that if $C_k \clP = C_{k+1}\clP$, then $\classfont{CH} = C_k \clP$. Indeed, we have 
$C_{k+2}\clP=\clPP^{C_{k+1} \clP}=\clPP^{C_k \clP}=C_{k+1} \clP$ and then proceed by induction.

\subsection{Rational interactive proofs}
Rational interactive proofs were introduced by P. Azar and S. Micali in~\cite{azar2012rational}. We will call the prover Merlin and we will call the verifier Arthur. Let us fix a finite alphabet $\Sigma$ (e.g. $\Sigma = \{0,1\}$) and assume we are working with languages over $\Sigma$.

We start from definitions of single-round protocols and the respective classes.
%Before proceeding with more general, and thus much more technical, definitions, let us give a definition of the first level of the $\clFRMA$-hierarchy, the class $\clFRMA[1]$, which is of interest by itself.

\begin{definition}
    The class $\clFRMA[1]$ (from Functional Rational Merlin-Arthur games with 1 round) is the set of all functions $f: \Sigma^* \to \Sigma^*$, for which there exist polynomial-time computable functions $R: \Sigma^* \to [0,1]$ and $\phi: \Sigma^* \to \Sigma^*$, such that for all $x \in \Sigma^*$ and for any 
    \[
        m^* \in \operatorname{Argmax}_m \mathbb{E}_{r} R(x; m, r),
    \]
    where $r$ is distributed uniformly on $\{0,1\}^{l(|x|)}$ for some polynomial $l$, we have
    \[
        f(x) = \phi(x; m^*).
    \]
\end{definition}

Here $m$ is the Merlin's hint given to Arthur, $r$ is Arthur's randomness, $R$ is the reward paid by Arthur to Merlin given as a finite binary fraction, and $\phi$ specifies how to use the hint for computing $f$. A message $m^*$ maximizing the expected reward must be a hint leading Arthur to the correct value of $f(x)$. 

If $f(x)$ has only binary values, it can be treated as a characteristic function of some predicate. It leads to a similar class of languages.

\begin{definition}
    The class $\clDRMA[1]$ (from Decisional RMA) is the set of all languages $S\subset\Sigma^*$ with characteristic functions $\mathbb{I}_S(x) = \mathbb{I}\{x \in S\}$ in $\clFRMA[1]$.
\end{definition}

Azar and Micali show explicitly or implicitly the following results:
\begin{theorem}[\cite{azar2012rational}]\label{thm:azar-micali-one-round}
    $\clSharpP\subset\clFRMA[1]$, $\clPP\subset\clDRMA[1]$, $\clParityP\subset\clDRMA[1]$.
\end{theorem}

Now we expand the definitions to multi-round protocols.

\begin{definition}
A $k$-round rational interactive protocol $\Pi$ consists of poly\-nomial-time computable functions ${V: \Sigma^* \to \Sigma^*}$, ${R: \Sigma^* \to [0,1]}$ and ${\phi: \Sigma^* \to \Sigma^*}$. For a string $x \in \Sigma^*$ the interaction in protocol $\Pi$ on input $x$ is defined as follows.
\begin{itemize}
    \item We denote Arthur's messages as $a_1, a_2, \ldots, a_k$ and Merlin's messages as $m_1, m_2, \ldots, m_k$. All messages are required to be of polynomial in $|x|$ length.
    \item Arthur's messages are computed as
    \[
        a_{i+1} = V(x; m_1, a_1, m_2,\ldots, a_i, m_{i+1}; r_1, \ldots, r_{i+1}),
    \]
    where $r_1, \ldots, r_k$ are Arthur's private random strings.
    \item We call $\phi$ the value function and we call the value $\phi(x; m_1, a_1, \ldots, m_k)$ the value computed in the protocol.
    \item We call $R$ the reward function and we call the value $R(x; m_1, r_1, \ldots, m_k, r_k)$ the reward.
\end{itemize}
\end{definition}

\begin{remark}
    The value computed in the protocol does not depend on the last message of Arthur, since this dependency can be included in $\phi$. We say that the reward depends only on Merlin's messages and Arthur's randomness, since Arthur's messages can be restored from this information.
\end{remark}

\begin{remark}
    In this paper, while describing protocols, we will talk about what messages Merlin sends. This is mostly needed for making a protocol's idea clearer to the reader and should be interpreted as a description of how Arthur processes received strings and what format is expected of Merlin (we assume that messages that deviate from described formats lead to $0$ reward).
\end{remark}

\begin{definition}
    In a $k$-round rational interactive protocol $\Pi$ with the reward function $R$, we say that Merlin is rational, if for all $i = 0, 1, \ldots, k-1$ it holds that
    \[
        m_{i+1} \in \operatorname{Argmax}_m \expect_{a_{i+1}} \max_{m_{i+2}}\expect_{a_{i+2}} \ldots\max_{m_k}\expect_{a_k} R(x; m_1, \ldots, a_i, m, a_{i+1},\ldots, m_k, a_k),
    \]
    i.e. a rational Merlin always chooses a message that maximizes the expected reward, conditional on the currently obtained information.
\end{definition}

\begin{definition}
    We say that a function $f$ is computed in a rational interactive protocol $\Pi$, if for all $x$ and for any Merlin's rational behavior it holds that
    \[
        f(x) = \varphi(x; m_1, a_1, \ldots, m_k).
    \]
\end{definition}

\begin{definition}
    Class $\clFRMA[k]$ is a set of all functions computable in some $k$-round interactive protocol.
\end{definition}

\begin{definition}
    Class $\clDRMA[k]$ is a class of languages $S$ such that their characteristic functions $\mathbb{I}_S(x) = \mathbb{I}\{x \in S\}$ are in $\clFRMA[k]$. When speaking about $\clDRMA[k]$-protocols we will often call a corresponding value function a decision function and denote it with $\pi$.
    
    By $\clDRMA$ we denote the union of all these classes:
    \[
        \clDRMA = \bigcup_{k\geq 0} \clDRMA[k].
    \]
\end{definition}

In the pioneering paper~\cite{azar2012rational}, the following connections between the counting hierarchy and the $\clDRMA$ hierarchy were established.
\begin{theorem}\label{thm:azar_micali_lower_bound}
    $C_k\clP \subseteq \clDRMA[k]$, for all $k \geq 0$.
\end{theorem}
\begin{theorem}\label{thm:azar_micali_upper_bound}
   $\clDRMA[k] \subseteq C_{2k+1}\clP$, for all $k \geq 0$. 
\end{theorem}
Together these two results show that $\classfont{CH} = \clDRMA$.

The following notation simplifies the reasoning about Merlin's rationality.

\begin{definition}
    Given a rational interactive protocol $\Pi$ and some transcript $\trT_i$, \\${\trT_i = (x; m_1, a_1, \ldots, s), s \in \{m_i, a_{i-1}\}}$ of the first $i$ rounds (not necessarily obtained in an interaction with a rational Merlin), denote by $E^{\Pi}_i(\trT_i)$  the maximal expected reward if Merlin behaves rationally in all future rounds.
\end{definition}

Finally, to construct and analyze protocols, we need the following notion.
\begin{definition}
    Given a rational interactive protocol $\Pi$ denote by $\Delta(\Pi)$ a poly\-nomial-time computable function that gives a positive lower bound for the difference between the expected rewards for an optimal and a second optimal Merlin's messages at the first time he deviates from the rational behavior, i.e. for all $x$
    \[
        0 < \Delta(\Pi)(x) \leq \min_{\substack{\textsl{$i, \trT_i$; where $\trT_i$ is obtained in}\\ \textsl{an interaction with a rational Merlin,}\\ m' \notin \operatorname{Argmax}_m E^{\Pi}_{i+1}( \trT_i, m)}} \left( \max_m E^{\Pi}_{i+1}( \trT_i, m) - E_{i+1}^{\Pi}(\trT_i, m')\right).
    \]
\end{definition}
\begin{remark}
    In the model we are working with, where the reward should be expressed as a finite binary fraction and Arthur's random bits are distributed uniformly and independently, the expected reward can be written as a sum of finite binary fractions (the reward value multiplied by its probability) of length polynomial in input size. This allows us to use $\Delta(\Pi)(x) = 2^{-p(|x|)}$ for some polynomial $p$ dependent on $\Pi$. 
    The polynomial $p(n)$ can be taken as $T_R(n) + l_r(n)$, where $T_R(n)$ is the maximal time for the reward to be computed on $x$ with $|x|=n$ (and hence the value of the reward can be expressed as $\frac{k}{2^{T_R(n)}}$ for an integer $k$), and $l_r(n)$ is the number of random bits Arthur uses on the input of length $n$.  
    In other models (e.g. the ones where the reward is expressed as a ratio of two integers) the corresponding assumptions about existence of a computable $\Delta(\Pi)$ and adjustments in theorem statements and proofs should be made. It should also be noted that such functions exist for the protocols in \cite{azar2012rational}.
\end{remark}

\begin{remark}
    Function $\Delta(\Pi)$ is not the same thing as the reward gap as it used, for example, in
    \cite{guo2014singleround}. 
\end{remark}

\section{Results concerning DRMA-hierarchy}\label{sect:DRMA}

We start with proving that Arthur can begin an interactive proof with one Merlin and continue it with the other. Moreover, the second Merlin can be asked the expected reward of the first Merlin at the moment of switching.

\begin{lemma}\label{thm:new_main_lemma}
    For any integer $i, j \geq 0$, for any $\clFRMA[i + j]$-protocol $\Pi$ for computing function $f$ there exists a function $\widetilde{f}$ that: 
    \begin{enumerate}
        \item for any $x$ maps any transcript $\trT_i$ of the first $i$ rounds of interacting with a rational Merlin according to protocol $\Pi$ on input $x$ to $(E^{\Pi}_{i}(\trT_i), f(x))$,
        \item maps any transcript $\trT_i$ of $i$ rounds (i.e. of interaction between Arthur and an arbitrary prover) to $(E_i^{\Pi}(\trT_i), v)$ for some $v$,
        \item is in $\clFRMA[j]$.
    \end{enumerate}
\end{lemma}

The idea of the proof is that for a sampleable random variable there is a protocol for determination of its expected value, and this protocol can be combined with the original protocol by scaling down the reward in the protocol for expectation. See the complete proof in the Appendix.

\begin{corollary}\label{thm:corol_from_main_1}
    For all $i, j \geq 0$, the following inclusions hold: 
    \[
        \clFRMA[i + j] \subseteq \clFRMA[i]^{(\clFRMA[j])[1]} \subseteq \clFRMA[i]^{||\clDRMA[j]},
    \]
    where {\lq\lq $[1]$\rq\rq} in the second class means that only one oracle query is allowed.
\end{corollary}
\begin{proof}
    The case of $i = 0$ or $j = 0$ is trivial.
    Let $\Pi$ be an $\clFRMA[i+j]$-protocol for some function $f \in \clFRMA[i+j]$. Consider a new protocol $\Pi'$, which coincides with $\Pi$ up to Merlin's message in the $i$-th round. Then Arthur computes his $i$-th message in protocol $\Pi$ and sends the current transcript $\trT_i$ to an $\clFRMA[j]$-oracle that corresponds to a function from lemma \ref{thm:new_main_lemma}.
    The value computed by Arthur and the reward are the ones obtained from the oracle.
    
    The rational Merlin will follow the protocol, because
    \begin{gather*}
        E_{i-1}^{\Pi'}(\trT_{i-1}) = \max_{m} \expect_{r} R\big((\trT_{i-1}, m, r)\big)=
        \max_{m} \expect_{r} E_i^{\Pi}\big((\trT_{i-1}, m, r)\big) = E_i^{\Pi}(\trT_i),
    \intertext{from which it follows by downward induction that for all $t < i$}
        E_{t}^{\Pi'}(\trT_{t}) = \max_m \expect_r E_{t+1}^{\Pi'}\big((\trT_{t}, m, r)\big) = \max_m \expect_r E_{t+1}^{\Pi}\big((\trT_{t}, m, r)\big) = E_{t}^{\Pi}(\trT_{t}),
    \end{gather*}
    meaning that the argmaxima, and so the strategies for rational Merlins, coincide for $\Pi$ and $\Pi'$. That means that we can treat the transcript $\trT_i$ from protocol $\Pi'$ as the transcript obtained from interaction with a rational Merlin in protocol $\Pi$, so by lemma \ref{thm:new_main_lemma} the oracle will also return the correct value of $f$.
    
    The second inclusion can be easily obtained by replacing the oracle function $\widetilde{f}$ with the language 
    \[
        A_{\widetilde f} = \left\{(x, i, b) \mid \textsl{$b$ is the $i$-th bit of $\widetilde{f}(x)$}\right\}
    \]
    and calculation of each bit of $\widetilde{f}(x)$ with parallel queries to the $A_{\widetilde f}$.
\end{proof}

\begin{corollary}\label{thm:corol_from_main_2}
    For all $i, j \geq 0$, the following inclusions hold: 
    \[
        \clDRMA[i + j] \subseteq \clDRMA[i]^{(\clFRMA[j])[1]} \subseteq \clDRMA[i]^{||\clDRMA[j]},
    \]
    where {\lq\lq $[1]$\rq\rq} in the second class means that only one oracle query is allowed.
\end{corollary}

\begin{proof}
    This trivially follows from corollary \ref{thm:corol_from_main_1} because a $\clDRMA[k]$ protocol is just a special case of an $\clFRMA[k]$ protocol. For the first inclusion, a $\clDRMA[i+j]$-protocol is an $\clFRMA[i+j]$-protocol for a binary-valued function $f$, for which, by corollary~\ref{thm:corol_from_main_1}, there exists a $\clFRMA[i]^{(\clFRMA[j])[1]}$-protocol that, being a protocol for a binary-valued function, constitutes a $\clDRMA[i]^{(\clFRMA[j])[1]}$-protocol for the corresponding language. The second inclusion can be proven similarly.
\end{proof}

Now we will show that upper bound of corollary \ref{thm:corol_from_main_2} is tight at least for $i = 0$. In fact, we prove a somewhat stronger result.
\begin{theorem}\label{thm:main_hierarchy}
    $ \clP^{\clFRMA[k]}=\clDRMA[k]$ for all $k \geq 0$.
\end{theorem}

The idea is, again, to scale down the reward for subproblems that depend on the answers to some other subproblems. With carefully performing proof by induction we can see that a fully rational Merlin has no incentive to lie in any of the subproblems. The complete proof can be found in the Appendix.

\begin{corollary}\label{crl:p-drma}
    $\clP^{\clDRMA[k]} = \clDRMA[k]$, for all $k \geq 0$.
\end{corollary}
\begin{proof}
    Since a $\clDRMA[k]$-protocol is a special case of an $\clFRMA[k]$-protocol, the inclusion $\clP^{\clDRMA[k]} \subseteq \clDRMA[k]$ follows from theorem \ref{thm:main_hierarchy}. 
    The inclusion in the other direction is trivial.
\end{proof}

This corollary shows a plausible difference between the $\clDRMA$ hierarchy and the counting hierarchy. For instance, $\clP^{\clPP}$ is known to contain $\clPH$ due to Toda's theorem, while for $\clPP$ it is unknown, then probably $\clPP\ne\clP^{\clPP}$.

\begin{theorem}\label{thm:pp_with_frma_oracle}
    $\clPP^{\clFRMA[k]} \subseteq \clDRMA[k + 1]$, for all $k \geq 0$.
\end{theorem}
\begin{proof}
    If $A$ is a language in $\clPP^{\clFRMA[k]}$, then there exists a language $B \in \clP^{\clFRMA[k]}$ and a polynomial $q$ such that for all $x$:
    \begin{itemize}
        \item $x \in A \Longrightarrow {2^{-q(|x|)}}\left|\left\{ y \in \{0,1\}^{q(|x|)} : (x, y) \in B \right\} \right| > \frac{1}{2}$,
        \item $x \notin A \Longrightarrow {2^{-q(|x|)}}\left|\left\{ y \in \{0,1\}^{q(|x|)} : (x, y) \in B \right\} \right| < \frac{1}{2}$.
    \end{itemize}
    By theorem \ref{thm:main_hierarchy}, there is a $\clDRMA[k]$-protocol $\Pi$ for $B$ with the reward function $R$ and decision function $\pi$. We now describe a $\clDRMA[k+1]$-protocol $\Pi'$ for the language $A$.
    \begin{enumerate}
        \item In the first round Merlin sends one bit $b$ that is supposed to equal $\mathbb{I}\{x \in A\}$.
        \item Then Arthur samples a random string $y \sim \mathcal{U}_{q(|x|)}$ and sends it to Merlin.
        \item The next $k$ rounds correspond to $k$ rounds of protocol $\Pi$ for string $(x, y)$.
        \item Finally, Arthur computes the value $\Delta = \Delta(\Pi)((x, y))$ as well as the reward $R$ and decision bit $\pi$ that correspond to the subprotocol for $(x, y)$, pays Merlin the reward $\frac{1}{2}R + \frac{\Delta}{4}\mathbb{I}\{b = \pi\}$ and accepts $x$ if and only if $b = 1$.  
    \end{enumerate}
    The rational Merlin will not deviate from the subprotocol $\Pi$, since deviating will decrease $\frac{1}{2}R$ by at least $\frac{1}{2}\Delta$ while increasing $\frac{\Delta}{4}\mathbb{I}\{b = \pi\}$ by at most $\frac{\Delta}{4}$. Hence the rational Merlin in $\Pi'$ will be replicating the actions of the rational Merlin in the protocol $\Pi$, in which case $\pi = \mathbb{I}\{(x, y) \in B\}$. So, to maximize the expected value of $\frac{\Delta}{4}\mathbb{I}\{b = \pi\}$, the rational Merlin will send $\mathbb{I}\{x \in A\}$ in the first round.
\end{proof}

Theorem \ref{thm:pp_with_frma_oracle} yields an independent proof of theorem \ref{thm:azar_micali_lower_bound} from \cite{azar2012rational}.
\begin{corollary}
    $C_k\clP \subseteq \clDRMA[k]$ for all $k \geq 0$.
\end{corollary}
\begin{proof}
    A straightforward induction by $k$.
\end{proof}

The next corollary shows    the connection between collapsing of the two hierarchies.

\begin{corollary}\label{thm:nice-corollary}
    If $\clDRMA[k]=\clDRMA[k+1]$, then the counting hierarchy collapses to the $(2k + 1)$-st level.
\end{corollary}
\begin{proof}
    We will prove by induction on $l \geq 0$ that $C_{k+l}\clP \subseteq \clDRMA[k]$. The base  case is provided by theorem \ref{thm:azar_micali_lower_bound}. Induction step is as follows:
    \[
        C_{k + (l+1)}\clP = \clPP^{C_{k + l}\clP} \subseteq \clPP^{\clDRMA[k]} \subseteq \clDRMA[k+1]=\clDRMA[k],
    \]
    where the last equality follows from the assumption of the corollary. 
    By theorem \ref{thm:azar_micali_upper_bound} we have $\clDRMA[k] \subseteq C_{2k + 1}\clP$, hence 
    \[
        C_{2k + 2}\clP = C_{k + (k + 2)}\clP \subseteq \clDRMA[k] \subseteq C_{2k+1}\clP,
    \]
    so $C_{2k+1}\clP = C_{2k+2}\clP$ and $\mathbf{CH} = C_{2k+1}\clP$.
\end{proof}

Note that we do not know whether $\clDRMA[k]=\clDRMA[k+1]$ implies $\clDRMA[k+1]=\clDRMA[k+2]$, but we still obtain from this corollary that $\clDRMA[k]=\clDRMA[k+1]$ implies the collapse of the $\clDRMA$ hierarchy at some higher level.

\section{Communication complexity of single-round proofs}\label{sect:comm}

In this section we study classes of the form $\clDRMA[1, c]$, that is, classes of languages that are recognized by $\clDRMA[1]$ protocols with at most $c$ bits sent by Merlin. We also assume that Merlin uses the binary alphabet.

Firstly, let us reduce to the case of binary rewards.

\begin{lemma}\label{thm:one_bit_reward}
    Every $\clDRMA[1, c]$-protocol is equivalent to a $\clDRMA[1, c]$-protocol with the reward function taking values in $\{0, 1\}$.
\end{lemma}
\begin{proof}
    Arthur should, after computing the reward $R$ in the given protocol, send as the reward not the value $R$, but the $i$-th bit of $R$ with probability $2^{-i}$ (which can be done in finite time since for all sufficiently large $i$ the $i$-th bit of $R$ would be $0$), this way Arthur would essentially send $1$ with probability $R$.
\end{proof}

Note that this transformation can be applied to any $\clDRMA$-protocol without increasing its round or communication complexity.

\begin{lemma} \label{lemma:compare_expects}
    For any two circuits $C_0, C_1$ which compute functions from $\{0,1\}^n$ to $\{0,1\}$, the problem of deciding whether or not $\expect_{x\sim {\cal U}_n} C_0(x) \leq \expect_{x\sim{\cal U}_n} C_1(x)$ is in $\clPP$. In other words, the language
    \begin{multline*}
        \mathsf{Compare-Expectations} = \Big\{(C_0, C_1)\mid C_i: \{0,1\}^n \to \{0,1\},\\ \expect_{r\sim {\cal U}_n} C_0(r) \leq \expect_{r\sim {\cal U}_n} C_1(r) \Big\}
    \end{multline*}
    is in $\clPP$.
\end{lemma}
Circuits are used here as a convenient way to describe parameterized functions. in our applications a polynomial-time computable function with a fixed parameter is represented by a polynomial-size circuit.
The idea of the proof is to sample both circuits and favor the one with greater value while breaking ties randomly. The proof is presented in the Appendix.

The following lemma is implicit in~\cite{azar2012rational}, but we prove it for completeness.
\begin{lemma}\label{thm:PP_in_DRMA_1_1}
    $\clPP \subseteq \clDRMA[1, 1]$.
\end{lemma}
\begin{proof}
    It is well-known that for any language $A \in \clPP$ there is a polynomial-time algorithm $V$ such that $\Pr_y [V(x, y) = 1]$ is strictly greater than $\frac{1}{2}$ for $x \in A$ and strictly less than $\frac{1}{2}$ for $x \notin A$.
    
    Consider the following protocol.
    \begin{enumerate}
        \item Merlin sends one bit $b$ that is supposed to be $\mathbb{I}\{x\in A\}$,
        \item Arthur runs $V(x, y)$ on random $y$ and pays Merlin $\mathbb{I}\{b = V(x, y)\}$, then accepts $x$ if and only if $b=1$.
    \end{enumerate}
    Since there are no ties, a rational Merlin would send the correct bit $b$.
\end{proof}

Finally, to prove the main result of this section, we use the following theorem by Fortnow and Reingold.

\begin{theorem}[\cite{fortnow1996pp}]\label{thm:Fortnow_PP_reductions}
    If $A \leq_{tt} B$ and $B \in \clPP$, then $A \in \clPP$.
\end{theorem}

We proceed by the main result of this section.

\begin{theorem}\label{thm:main_comm_thm}
    $\clDRMA[1, O(\log n)] \subseteq \clPP$.
\end{theorem}
\begin{proof}
    Let $A$ be a language in $\clDRMA[1, O(\log n)]$ and $\Pi$ be the corresponding protocol for $A$ with the reward function $R$ and the decision function $\pi$. Since Merlin sends only a message of length $O(\log n)$, there can only be $p(n) = 2^{O(\log n)} = \poly(n)$ possible messages. 
    
    Consider an algorithm that, on input $x$, looks through all possible Merlin's messages and finds the message $m^*$ that maximizes the value $\expect_r R(x; m, r)$. This winner can be found in a knockout tournament where messages $m_0$ and $m_1$ are compared by feeding the pair $(R(m_0, \cdot), R(m_1, \cdot))$ to the $\mathsf{Compare-Expectations}$ oracle. Finally, $x$ is accepted iff $\pi(m^*) = 1$. Clearly, this algorithm provides a polynomial-time truth-table reduction from $A$ to $\mathsf{Compare-Expectations} \in \clPP$. Hence, by theorem \ref{thm:Fortnow_PP_reductions}, we have $A \in \clPP$.
\end{proof}

\begin{corollary}\label{crl:pp_eq_drma_log}
    $\clPP = \clDRMA[1, O(\log n)]$.
\end{corollary}
\begin{proof}
    Immediately follows from lemma \ref{thm:PP_in_DRMA_1_1} and theorem \ref{thm:main_comm_thm}.
\end{proof}

Let us now discuss lower bounds on communication complexity of $\clParityP$ for $\clDRMA[1]$-protocol. Corollary \ref{crl:pp_eq_drma_log} tells that $\clParityP \not\subseteq \clDRMA[1, O(\log n)]$ would imply $\clParityP \neq \clPP$, while it holds that $\clParityP \subseteq \clDRMA[1, \poly(n)]$ (Theorem~\ref{thm:azar-micali-one-round}), so it is interesting to study the communication complexity of $\clDRMA[1]$-protocols for problems in $\clParityP$. In all presented proofs so far we have treated certificate verification algorithms that define inclusion in counting classes as samplers from Bernoulli distributions. We now show that this approach cannot yield $\clDRMA[1]$-protocols with sublinear communication complexity.

\begin{theorem}\label{thm:parity_p_lower_bound}
    There is no polynomial-time computable functions $R: \{0,1\}^* \to [0,1]$ and $\phi: \{0,1\}^* \to \{0,1\}$ such that for any $n$ and $0 \leq k \leq 2^n$ the value of $\phi(1^n, m^*)$ is equal to the parity of $k$ where $m^*\in\operatorname{Argmax}_m \mathbb{E}_{a, r} R(1^n, m, a, r)$,
    $|m^*| < \alpha n$ for some constant $\alpha \in (0, 1)$, $a \sim {\cal U}_{s}$ with $s = \poly(n)$ and $r \in \{0,1\}^{d}$ with $d = \poly(n)$ and bits of $r$ being independently sampled from $\operatorname{Bern}\left(\frac{k}{2^n}\right)$.
\end{theorem}
The idea is that, for fixed $n, m$, the reward function is a polynomial in probability of nonuniform bits being equal to $1$. If $m$ is short, then some two of these polynomials of polynomial degree have to intersect in exponentially many points. The details are presented in the Appendix.

\section{Further research}\label{sect:concl}

In this paper, we show several interesting connections between the counting hierarchy and the $\mathbf{DRMA}$ hierarchy. But the main question is open: do they coincide or are they different? In case they are different a natural next step would be to construct an example of a language that lies in $\mathbf{DRMA}[k]$ but seems to not lie in $C_k\mathbf{P}$.

\bibliographystyle{splncs04}
\bibliography{drma}

\appendix
\section{Proof of lemma \ref{thm:new_main_lemma}}
\begin{proof}
    Denote by $R$ the reward function in $\Pi$ and by $\phi$ the value function there. The case of $i = 0$ or $j = 0$ is trivial, so we assume $i, j \geq 1$.
    Consider the following protocol $\widetilde{\Pi}$: in the first round,
    given the input $x$ and a partial transcript~$\trT$,
    Merlin sends ${E = E^{\Pi}_i(\trT)}$ together with $(i+1)$-st message that Merlin would send in protocol $\Pi$, and all other messages are exactly the messages sent by a rational Merlin
    in $\Pi$. In the end Arthur pays Merlin 
    \[
        \widetilde{R} = \frac{1}{2}R(\trT) + \frac{1}{4}\Delta(\Pi)(x)\Big( 1 - \big( E - R(\trT) \big)^2\Big),
    \]
    or he pays $0$ if Merlin's message format does not correspond to the format prescribed by the protocol (e.g. if the value $E$ is not in the range $[0, 1]$). Note that $\widetilde R$ is polynomial-time computable since $R$ and $\Delta(\Pi)(\cdot)$ are polynomial-time computable by definition (we can use a function of the form $2^{-\poly(|x|)} for \Delta(\Pi)$). 
    Since the reward is a value in $[0, 1]$, the value of $(E - R(\trT))^2$ also lies in $[0,1]$, which means that, if Merlin ever deviates from the {\lq\lq continuation\rq\rq} of protocol $\Pi$ at some step, then the expectation of $\frac{1}{2}R(\trT)$ will decrease at least by $\frac{1}{2}\Delta(\Pi)(x)$, while the expected value of $\frac{1}{4}\Delta(\Pi)(x)\Big( 1 - \big( E - R(\trT) \big)^2\Big)$ will not increase by more than $\frac{1}{4}\Delta(\Pi)(x)$. That means that the rational Merlin may only deviate from the protocol by sending an incorrect $E$ in the first round. 
    
    But then he will need to choose $E$ to minimize the functional
    \[
        \expect \Big[(E - R(\trT))^2 \big| \trT_i \sqsubseteq \trT \Big],
    \]
    where $\trT_i \sqsubseteq \trT$ means that $\trT$ is a continuation of $\trT_i$. This value is minimized by 
    \[
        E = \expect\big[R(\trT) \mid \trT \sqsubseteq \trT_i \big] = E^{\Pi}_i(\trT_i),
    \]
    where the last equality holds due to the fact that, as we have established, the rational Merlin in protocol $\widetilde{\Pi}$ will behave as the rational Merlin in protocol $\Pi$. 
    
    Finally, Arthur returns the value $\widetilde{\varphi} = (E, \varphi(\trT'))$, where $\trT'$ is just $\trT$ without the mentioning of $E$. From the established behavior of a rational Merlin it immediately follows that $E = E^{\Pi}_i(\trT_i)$ and also, if $\trT_i$ is produced in the interaction with a rational Merlin, then $\varphi(\trT') = f(x)$.
\end{proof}

\section{Proof of the theorem \ref{thm:main_hierarchy}}

\begin{proof}
    The inclusion $\clDRMA[k] \subseteq \clP^{\clFRMA[k]}$ is trivial.
    
    To prove the other inclusion, let $A$ be a language in $\clP^f$ for some ${f \in \clFRMA[k]}$, and let $\Pi$ be a $\clDRMA[k]$-protocol for $f$. 
    Also let $M$ be a polynomial-time Turing machine with access to oracle $f$ that recognizes $A$. Now consider the following $k$-round protocol $\widetilde \Pi$ for $A^f$ for input string $x$.
    \begin{enumerate}
        \item In the first round Merlin sends a message $m_1$, containing a number $l$, block of $l$ strings $y_1, y_2, \ldots, y_l$ and block of $l$ strings $m_1^{(1)}, m_1^{(2)}, \ldots, m_1^{(l)}$:
        \[
            m_1 = \left(l; y_1, \ldots, y_l;m_1^{(1)},\ldots, m_1^{(l)}\right).
        \]
        \item Arthur then emulates the machine $M$ on input $x$, using strings $y_i$ as answers to the oracle queries $x_1, \ldots, x_l$, and interprets strings $m_1^{(j)}$ as first messages of Merlin in the protocol $\Pi$ for $x_i$ (if emulation of $M$ makes not exactly $l$ queries, the final reward is set to be $0$).
        \item During all the following rounds, Merlin sends messages of the format $m_i = \left(m_i^{(1)}, \ldots, m_i^{(l)}\right)$, and Arthur sends messages of the format $a_i = \left(a_i^{(1)}, \ldots, a_i^{(l)}\right)$, where Arthur computes each $a_i^{(j)}$ as a response to transcript $$\trT_i^{(j)} = \left(x_j; m_1^{(j)}, a_1^{(j)}, \ldots, m_i^{(j)}\right),$$ and Merlin is expected to behave similarly.
        \item After all $k$ rounds, Arthur computes $\Delta = \min_j \Delta(\Pi)(x_j)$. and rewards $R_j$ for each {\lq\lq subprotocol\rq\rq}, and, if Merlin has not deviated from the protocol in an obvious way, pays him
        \[
            R = \frac{1}{2}R_1 + \frac{\Delta}{4}R_2 + \ldots + \frac{\Delta^{l-1}}{2^l}R_l + \frac{\Delta^l}{2^{l+1}}.
        \]
        Arthur then accepts $x$ if and only if it is accepted by the machine $M$ that gets the strings $y_j$ as answers to queries $x_j$ and in each {\lq\lq subprotocol\rq\rq} Merlin has proven that $f(x_j) = y_j$.
    \end{enumerate}
    
    Let us show that a polynomial-time randomized Arthur is capable of following the described protocol. In the protocol, the number of queries and the length of each query $x_j$ are bounded by the time the machine $M$ works on input $x$, hence for each $j$ we have $|x_j| = \poly(|x|)$ and $l = \poly(|x|)$. That means that each $y_j$, each message $a_i^{(j)}$, $m_i^{(j)}$ and each composed message $a_i, m_i$ is also of length $\poly(|x|)$.  Values $\Delta(\Pi)(x_j)$ are polynomial-time computable values of polynomial length by definition. Computation of $R$ requires a polynomial number of basic arithmetical operations with polynomial length rational numbers, so it can be performed in polynomial time as well.
    
    Now we show that a rational Merlin will not deviate from the protocol. Since $R$ is a strictly positive value, a rational Merlin has no incentive to violate the message format or provide $l$ that does not correspond to the number of queries in simulation of $M$ (which would lead to $0$ reward).
    
    By induction on $j$ we will prove that for all $j$:
    \begin{enumerate}
        \item[\textbf{(a)}] the $j$-th query $x_j$ of the simulation of $M$ is equal to the $j$-th query of the actual machine $M$ with access to $f$,
        \item[\textbf{(b)}] for all $i=1,2,\ldots,k$, the string $m_i^{(j)}$ is optimal for a rational Merlin in the corresponding {\lq\lq subprotocol\rq\rq},
        \item[\textbf{(c)}] $y_j = f(x_j)$.
    \end{enumerate}
    
    To prove {\bf (a)} for $j$ assuming everything is proven for all $j' < j$, notice that {\bf (a)} and {\bf (c)} for all $j' < j$ mean that the simulation of $M$, since it is a deterministic Turing machine, works in the same way as $M^f$ up until the $j'$-th query to the oracle, including that query.
    
    Now we prove {\bf (b)} for $j$ assuming it is proven for all $j' < j$. Assume that a rational Merlin did not deviate in the $j$-th subprotocol until the $i$-th round. Let $\widetilde m^{(j)}_i$ be an optimal $i$-th message in $j$-th subprotocol, and $\widehat m^{(j)}_i$ be a message suboptimal for the $j$-th subprotocol, let $\widehat m_i$ be the corresponding dishonest message. By the induction hypothesis a rational Merlin would choose optimal messages for all previous subprotocols, so
    \begin{align*}
        E_{i-1}^{\widetilde \Pi} (\trT_{i-1}) 
            &\geq 
        \sum_{j' < j} \frac{\Delta^{j'-1}}{2^{j'}}E_{i-1}^{\Pi}\left(\trT_{i-1}^{(j')}\right) + \frac{\Delta^{j-1}}{2^j}E_i^{\Pi}\left(\left(\trT_{i-1}^{(j)}, \widetilde{m}_i^{(j)}\right)\right) 
            \\&\geq 
        \sum_{j' < j} \frac{\Delta^{j'-1}}{2^{j'}}E_{i-1}^{\Pi}\left(\trT_{i-1}^{(j')}\right) + \frac{\Delta^{j-1}}{2^j}\left(E_i^{\Pi}\left(\left(\trT_{i-1}^{(j)}, \widehat{m}_i^{(j)}\right)\right) + \Delta(\Pi)(x_j)\right)
            \\&\geq
        \sum_{j' < j} \frac{\Delta^{j'-1}}{2^{j'}}E_{i-1}^{\Pi}\left(\trT_{i-1}^{(j')}\right) + \frac{\Delta^{j-1}}{2^j}E_i^{\Pi}\left(\left(\trT_i^{(j)}, \widehat{m}_i^{(j)}\right)\right) + \frac{\Delta^{j}}{2^j}
            \\&>
        \sum_{j' < j} \frac{\Delta^{j'-1}}{2^{j'}}E_{i-1}^{\Pi}\left(\trT_{i-1}^{(j')}\right) + \frac{\Delta^{j-1}}{2^j}E_i^{\Pi}\left(\left(\trT_i^{(j)}, \widehat{m}_i^{(j)}\right)\right) + \sum_{j < j' \leq l+1} \frac{\Delta^{j'-1}}{2^{j'}}
        \\&\geq 
        E_i^{\widetilde \Pi}\left(\trT_i, \widehat m_i\right),
    \end{align*}
    here all inequalities are due to the definitions, and the last inequality follows from the induction hypothesis for {\bf (b)} and an upper bound on geometric series (applying that all the values $R_j'$ do not exceed $1$): the expected reward for $\left(\trT, \widehat m_i\right)$ is a combination of the expected rewards for partial transcripts $\left(\trT_{i-1}^{(j')}\right)$ for $j' < j$ by the induction hypothesis for \textbf{(b)}, the expected reward for the partial transcript $\left(\trT_i^{(j)}, \widehat{m}_i^{(j)}\right)$, and the expected rewards for the rest of the subprotocols, which do not exceed $1$.
    
    Finally, {\bf (c)} for $j$ follows from {\bf (b)} for $j$, since a rational Merlin in $\widetilde \Pi$ behaves in the $j$-th subprotocol like a rational Merlin in $\Pi$.
\end{proof}

\section{Proof of the lemma \ref{lemma:compare_expects}}

\begin{proof}
    We will describe the probabilistic polynomial time algorithm that decides the language $\mathsf{Compare-Expectations}$ with probability at least $1/2$. 
    
    On input $(C_0, C_1)$ the algorithm  independently samples two uniformly distributed strings $r, r'$ of length $n$ and computes two bits: $c_0 = C_0(r)$  and bit ${c_1 = C_1(r')}$. If $c_0 < c_1$ (i.e., $c_0 = 0$ and $c_1 = 1$), then the algorithm accepts the input, if $c_0 > c_1$ --- rejects. If $c_0 = c_1$, the algorithm accepts with probability $1/2$ and otherwise rejects the input. 
    
    For $i=0,1$, denote by $E_i$ the value $\expect_{r\sim{\cal U}_n}C_i(r)$. Then
    \begin{gather*}
        \Pr[\textsl{ the algorithm accepts }] = 
        \Pr[c_0 < c_1] + \frac{1}{2}\Pr[c_0 = c_1] =\\
        = \Pr[c_0 = 0, c_1 = 1] + \frac{1}{2}\big(\Pr[c_0 = c_1 = 0] + \Pr[c_0 = c_1 = 1]\big) =\\
        = (1 - E_0)E_1 +
        \frac{1}{2}\big((1-E_0)(1 - E_1) + E_0E_1\big) = \\
        = \frac{1}{2} + \frac{E_1 - E_0}{2},
    \end{gather*}
    which is at least $\frac{1}{2}$ if and only if $(C_0, C_1) \in \mathsf{Compare-Expectations}$. The algorithm terminates in polynomial time since we assume that the size of a scheme is at least the size of its input.
\end{proof}

\section{Proof of theorem \ref{thm:parity_p_lower_bound}}

\begin{proof}
    Assume that such functions $R$ and $\phi$ and a constant $\alpha < 1/2$ exist. For each $m$, consider the polynomial
    \[
        Q_m(p) = \frac{1}{2^s}\sum_{a \in \{0,1\}^s} \sum_{r \in \{0,1\}^d} R(1^n, m, a, r) p^{|r|_1} (1-p)^{|r|_0},
    \]
    where $|r|_b$ is the number of bits of $r$ that are equal to $b \in \{0,1\}$. Note that for all $m$ we have $\deg Q_m \leq |r|_1 + |r|_0 = |r| = d = \poly(n)$.
    Clearly, for all $m$ we also have $\mathbb{E}_{a, r} R(1^n, m, a, r) = Q_m\left(\frac{k}{2^n}\right)$. 
    
    By our assumption, there exists a (not necessarily computable) function ${m_*: \{0, 1, \ldots, 2^n\} \to \{0,1\}^{< \alpha n}}$ such that for each $k$:
    \begin{enumerate}
        \item $\phi(1^n, m_*(k)) = (k \bmod 2)$,
        \item $Q_{m_*(k)}\left(\frac{k}{2^n}\right) > Q_{\widetilde m}\left(\frac{k}{2^n}\right)$ for any $\widetilde m$ for which $\phi(1^n, \widetilde m) \neq (k \bmod 2)$.
    \end{enumerate}
    
    By pigeonhole principle there exist a string $m_0, m_1$ such that $m_0 = m_*(k)$ for at least $\frac{2^n}{2^{\alpha n + 1}}$ even values of $k$. Since these numbers are all even, there are at least $\frac{2^n}{2^{\alpha n + 1}} - 1$ odd numbers interlaced with $m_*^{-1}(m_0)$. By applying the pigeonhole principle again, we have that there is a string $m_1$ such that there are at least $\frac{\frac{2^n}{2^{\alpha n + 1}} - 1}{2^{\alpha n + 1}} > 2^{(1 - 2\alpha) n - 3}$ odd numbers in $m_*^{-1}(m_1)$ that are interlaced with some subset of $m_*^{-1}(m_0)$. Let $k_0 < k_1 < \ldots < k_N$ denote the interlaced subsets of $m_*^{-1}(m_0)$ and $m_*^{-1}(m_1)$, where $k_i \in m_*^{-1}(m_0)$ for even values of $i$ and $k_i \in m_*^{-1}(m_1)$ for odd values of $i$. 
    
    Let $\widehat Q(p)$ denote the difference $Q_{m_0}(p) - Q_{m_1}(p)$. Observe that 
    $\widehat Q(k_i) >0 $ for all even $i$ and $\widehat Q(k_i) < 0$ for all odd $i$, and by the intermediate value theorem we have at least $N-1$ distinct zeros of $Q$. Since $N = \Omega(2^{(1-2\alpha)n})$ while $\deg \widehat Q \leq d = \poly(n)$, we have $\widehat Q \equiv 0$, while for all $i$ we have $\widehat Q (k_i) \neq 0$.
    The obtained contradiction yields the claimed result
\end{proof}

\end{document}